\documentclass[twocolumn,10pt]{IEEEtran}
\usepackage{braket}

\usepackage{booktabs}

\input{def.tex}
\usepackage{dsfont}
\usepackage{minibox}
\DeclareSymbolFont{matha}{OML}{txmi}{m}{it}
\DeclareMathSymbol{\varv}{\mathord}{matha}{118}
\usepackage{eurosym}
\usepackage{hhline,physics}

\usepackage{multicol}
\usepackage{algorithm}
\usepackage{graphicx}
\usepackage{textcomp}
\usepackage{caption,pifont}
\usepackage[multiple]{footmisc}

\usepackage{algorithmic}

\makeatletter
\renewcommand{\baselinestretch}{0.91}

\IEEEoverridecommandlockouts

\begin{document}
	\title{Stacked Intelligent Metasurface–Assisted Fluid Antenna Systems: Outage Probability} 
	\author{Anastasios Papazafeiropoulos,~\IEEEmembership{Senior Member,~IEEE}     \thanks{A. Papazafeiropoulos is with the Communications and Intelligent Systems Research Group, University of Hertfordshire, Hatfield AL10 9AB, U. K. 
			Corresponding author's email: tapapazaf@gmail.com. }}
	\maketitle\vspace{-1.9cm}
	\begin{abstract}	
Stacked intelligent metasurfaces (SIMs) and fluid antenna systems (FAS) are emerging technologies for wave-domain and spatial signal manipulation, respectively.This letter proposes a novel joint SIM–FAS communication model in which transmission and reception are performed by a SIM and an FAS, respectively. Using the block-diagonal matrix approximation (BDMA), a closed-form expression for the outage probability is derived, and the SIM phase shifts are optimized to minimize outage. Numerical results validate the analytical accuracy and demonstrate substantial performance gains over conventional benchmark schemes.

	\end{abstract}
	
	
	\section{Introduction}
	
Stacked intelligent metasurfaces (SIMs) extend conventional metasurfaces by enabling wave-domain analog beamforming with reduced hardware complexity and energy consumption~\cite{An2023d,Papazafeiropoulos2024a}. A SIM-assisted holographic MIMO framework with few RF chains was introduced in~\cite{An2023d}, while subsequent works investigated hybrid SIM-based MISO/MIMO designs and achievable rate performance in multi-user systems~\cite{Papazafeiropoulos2024a,Papazafeiropoulos2025,Papazafeiropoulos2024b}.\footnote{\textcolor{black}{In addition, SIMs improve energy efficiency by enabling wave-domain analog processing with fewer RF chains, reducing hardware complexity and power consumption. Optimal SIM design requires balancing passive beamforming gains with implementation costs, where layer count and system architecture critically determine the overall energy efficiency trade-off \cite{Shi2026,Perovic2025}.}}
	
Fluid antenna systems (FASs) enable dynamic repositioning of the radiating port, introducing an additional spatial degree of freedom for physical-layer optimization~\cite{Wong2023,Wong2020,Khammassi2023,VegaSanchez2023}. Their feasibility has been experimentally validated~\cite{Shen2024}, and extensions to multi-antenna and port-selection frameworks have been widely studied~\cite{New2023,Chai2022,Mei2024}. Surveys and recent works have also examined FAS integration with reconfigurable intelligent surfaces (RISs)~\cite{Yao2025,Yao2025a}.

This work proposes a novel SIM–FAS communication framework based on the block-diagonal matrix approximation (BDMA) method~\cite{RamirezEspinosa2024}. By combining wave-domain reconfigurability with spatial adaptability, the proposed system integrates a SIM into the transmit chain, enabling analog precoding and coherent wavefront synthesis. This directly controls the field seen by the FAS ports, leading to distinct outage behavior compared to RIS-assisted FAS systems. 
\textcolor{black}{Unlike fluid reconfigurable intelligent surface (FRIS), which rely on a cascaded channel, the proposed SIM–FAS architecture performs wave-domain precoding at the transmitter, directly shaping both deterministic and scattered components \cite{Zhu2025}.  Furthermore, integrating FAS introduces spatial diversity via port selection, creating a joint wave- and spatial-domain design.}\footnote{\textcolor{black}{Recent FRIS introduces element-positioning flexibility to shape the propagation environment, whereas the proposed SIM–FAS framework performs transmit-side wave-domain precoding and receiver-side port selection, resulting in a fundamentally different signal model and statistics \cite{Zhu2025}. Furthermore, in practice, SIMs are affected by hardware impairments  \cite{Li2025}. However, in this work, we adopt an idealized model for analytical insight.}} Moreover, we derive a closed-form outage probability expression and optimize the SIM phase shifts accordingly. Numerical results validate the BDMA-based analysis and demonstrate significant performance gains over conventional benchmarks.

	\section{System Model}\label{System}
We consider a downlink where a base station (BS) employs a SIM for wave-domain precoding. The SIM has $L$ layers with $M$ meta-atoms, each described by $\boldsymbol{\Theta}_\ell = \mathrm{diag}(e^{j\theta_{\ell,1}},\ldots,e^{j\theta_{\ell,M}})$, $\ell \in \mathcal{L}$. The resulting transformation is
\begin{equation}
	\mathbf{g}_{\mathrm{SIM}} = \left( \prod_{\ell=2}^{L} \boldsymbol{\Theta}_{\ell}\mathbf{W}_{\ell} \right)\boldsymbol{\Theta}_{1}\mathbf{w}_{1},
	\label{eq:SIM}
\end{equation}
where $\mathbf{W}_{\ell} \in \mathbb{C}^{M \times M}$ is the inter-layer coupling matrix \cite[Eq.~9]{An2023d}, and $\mathbf{w}_{1}$ is the input excitation. The phases $\{\boldsymbol{\Theta}_\ell\}$ are set for constructive focusing.\footnote{\textcolor{black}{In this work, SIM phases are designed based on  statistical CSI and kept quasi-static, capturing large-scale effects, while the FAS exploits instantaneous fading via port selection for diversity.}}

At the receiver, the user employs a FAS with $N$ reconfigurable ports over an aperture $W\lambda$. The antenna selects the best port based on instantaneous channel conditions, exploiting spatial diversity along with the SIM beamforming gain.
	
We consider a Rician fading channel with line-of-sight (LoS) and non-line-of-sight (NLoS) components. Let $\mathbf{h}_k \in \mathbb{C}^{M \times 1}$ denote the channel to the $k$-th FAS port:
\begin{equation}
	\mathbf{h}_k = 
	\sqrt{\frac{\alpha K}{K+1}}\,\bar{\mathbf{h}}_k +
	\sqrt{\frac{\alpha}{K+1}}\,\tilde{\mathbf{h}}_k,
	\label{eq:h_k}
\end{equation}
where $\alpha$ is the path loss, $K$ the Rician factor, $\bar{\mathbf{h}}_k$ the LoS component, and $\tilde{\mathbf{h}}_k \sim \mathcal{CN}(0,1)$ the NLoS component.

	Given that the FAS ports are closely spaced, their channels exhibit spatial correlation. According to the Jakes model, the correlation coefficient between the first and the $k$-th port is given by
	\begin{equation}
		\mu_{1,k} = J_0\!\left(\frac{2\pi (k-1)W}{N-1}\right), 
		\quad k \in \{1,2,\dots,N\},
		\label{eq:jakes}
	\end{equation}
	where $J_0(\cdot)$ is the zero-order Bessel function of the first kind.  
	The correlation matrix of $\big[\mathbf{h}_1(m),\dots,\mathbf{h}_N(m)\big]$, where $m$ denotes the meta-atom index, is modeled as a Toeplitz matrix
	\begin{equation}
		\boldsymbol{\Sigma} =
		\begin{pmatrix}
			\mu_{1,1} & \mu_{1,2} & \cdots & \mu_{1,N}\\
			\mu_{1,2} & \mu_{1,1} & \cdots & \mu_{1,N-1}\\
			\vdots & \vdots & \ddots & \vdots\\
			\mu_{1,N} & \mu_{1,N-1} & \cdots & \mu_{1,1}
		\end{pmatrix} \in \mathbb{R}^{N \times N}.
		\label{eq:toeplitz}
	\end{equation}
	
	To balance analytical tractability and modeling accuracy, we adopt the BDMA to represent the spatial correlation among the $N$ ports.     \textcolor{black}{The BDMA approximation captures the dominant spatial correlation by grouping ports into locally correlated blocks, which is suitable when correlation varies smoothly across the FAS aperture (e.g., Jakes model). While more general non-stationary models offer higher fidelity, they significantly increase analytical complexity. BDMA thus provides a tractable balance between accuracy and insight, and the framework can be extended to more general models such as  \cite{Zhu2025} at the expense of tractability.}  In this approach, we have
	\begin{equation}
		\hat{\boldsymbol{\Sigma}} = \mathrm{Blkdiag}\!\left(\mathbf{C}_1, \dots, \mathbf{C}_B\right),
		\label{eq:bdma}
	\end{equation}
	where $\mathrm{Blkdiag}(\cdot)$ denotes the block-diagonal operator, and $\mathbf{C}_b$ is the correlation matrix of the $b$-th block given by
	\begin{equation}
		\mathbf{C}_b = \mathrm{toeplitz}\!\left(1, \mu_b^2, \dots, \mu_b^2\right) \in \mathbb{R}^{L_b \times L_b},
		\label{eq:block}
	\end{equation}
	subject to $\sum_{b=1}^{B} L_b = N$. Each block groups $L_b$ correlated ports characterized by an intra-block correlation coefficient $\mu_b$ determined via the BDMA algorithm in~\cite{An2023}.
	
	For the $L_b$ ports within the $b$-th block, the equivalent channel of the $k$-th port $(k \in \{1,\dots,L_b\})$ can be described as
	\begin{equation}
		\mathbf{h}_k =
		\sqrt{\frac{\alpha K}{K+1}}\,\bar{\mathbf{h}} +
		\mu_b \tilde{\mathbf{h}}_b +
		\sqrt{1-\mu_b^{2}}\,\mathbf{e}_k,
		\label{eq:h_equiv}
	\end{equation}
	where $\tilde{\mathbf{h}}_b$ and $\mathbf{e}_k$ are independent random vectors following $\mathcal{CN}\!\left(0, \frac{\alpha}{K+1}\mathbf{I}_M\right)$.\footnote{\textcolor{black}{Physically, BDMA approximates the spatial correlation induced by the finite FAS aperture $W\lambda$, where closely spaced ports exhibit strong correlation according to the Jakes model. As the correlation decays with spatial separation (via the Bessel function), ports within a localized region of the aperture experience similar fading and can be grouped into $B$ blocks of $L_b$ ports with approximately homogeneous intra-block correlation $\mu_b$. Inter-block correlation is neglected for tractability. Hence, the block structure reflects the spatial sampling of the aperture and the correlation decay across it.}}

	\subsection{Received Signal at the $k$-th Port}
	
	Let $x$ be the transmit symbol with $\mathbb{E}[|x|^2]=1$, and let $P$ denote the BS transmit power, then, the received baseband signal at the $k$-th FAS port is
	\begin{equation}
		y_k \;=\; \sqrt{P}\;\mathbf{h}_k^{\T}\,\mathbf{g}_{\mathrm{SIM}}\; x \;+\; n_k,
		\label{eq:yk}
	\end{equation}
	where $n_k\sim\mathcal{CN}(0,\sigma^2)$ is additive white Gaussian noise. 
	
	For convenience, define the  complex gain $C_k \;\triangleq\; \mathbf{h}_k^{\T}\,\mathbf{g}_{\mathrm{SIM}}$
	so that \eqref{eq:yk} becomes $y_k=\sqrt{P}\,C_k\,x+n_k$, and the instantaneous SNR at the $k$-th port is $	\gamma_k \;=\; \frac{P\,|C_k|^2}{\sigma^2}.$
	
	With port selection, the FAS activates $		k^{\star} \;=\; \arg\max_{k\in\{1,\dots,N\}} |C_k| ,$
		yielding the post-selection signal-to-noise ratio (SNR)  $		\gamma \;=\; \frac{P\,\max_k |C_k|^2}{\sigma^2}.$

	\section{Outage Probability }

	Building on the formulation provided in the previous section, if the target transmission rate is denoted by $R$, 
	the outage probability at the user can be expressed as
	\begin{equation}
		P_{\text{out}} = \Pr\!\big[\log_2(1+\gamma) < R\big].
	\end{equation}
	
	The proposed framework enables tractable analysis of the received SNR using only long-term channel statistics. The SIM provides wave-domain precoding at the transmitter, while the fluid antenna system enhances spatial diversity at the receiver, forming the basis for outage characterization and optimization of the SIM–FAS system.

With only statistical channel state information (CSI) available, the SIM phase shifts $\theta_{\ell,m}$ are assumed quasi-static. For the $b$-th BDMA block, the effective channel between the SIM-assisted BS and the $k$-th fluid-antenna port is given by	
	\begin{align}
		\mathbf{h}_k^{\T}\,\mathbf{g}_{\mathrm{SIM}}
		&	= 
		\sqrt{\frac{\alpha K}{K+1}}\,
		\bar{\bh}_{k}^{\T}\mathbf{g}_{\mathrm{SIM}}
		+ 
		\mu_{b}\tilde{\bh}_{b}^{\T}\mathbf{g}_{\mathrm{SIM}}\nn\\
		&	+ 
		\sqrt{1-\mu_{b}^{2}}\,\bee_{k}^{\T}\mathbf{g}_{\mathrm{SIM}}
		\label{eq:effective_channel}.
	\end{align}
	
	By defining $Z_{1}=\tilde{\bh}_{b}^{\T}\mathbf{g}_{\mathrm{SIM}}$, it follows that $Z_{1}$ is a linear transformation of a zero-mean 
	circular complex Gaussian vector. Consequently, $Z_{1}$ also follows a complex Gaussian distribution
	\begin{equation}
		Z_{1} \sim \mathcal{CN}(0,\tilde{\sigma}^{2}),
		\qquad 
		\tilde{\sigma}^{2} = \frac{\alpha}{K+1}\,\|\mathbf{g}_{\mathrm{SIM}}\|^{2}_{2}.
	\end{equation}
	
	Similarly, we define $Z_{2}=\bee_{k}^{\T}\mathbf{g}_{\mathrm{SIM}}$ with $Z_{2} \sim \mathcal{CN}(0,\tilde{\sigma}^{2})$ having the same variance with $Z_{1}$.
	
	Let $\delta = \sqrt{\frac{\alpha K}{K+1}}\,
	\bar{\bh}_{k}^{\T}\mathbf{g}_{\mathrm{SIM}}$ denote the deterministic 
	line-of-sight (LoS) component.  
	Conditioned on $\tilde{\bh}_{b}^{\T}\mathbf{g}_{\mathrm{SIM}}$, the envelope 
	$C_{k} = |\mathbf{h}_k^{\T}\,\mathbf{g}_{\mathrm{SIM}}|$ follows the distribution
	\begin{align}
		f_{C_{k}|\Delta_{b}}(r_{k}|r_{b})
		&=
		\frac{2r_{k}}{\tilde{\sigma}^{2}(1-\mu_{b}^{2})}
		\exp\!\left[-\frac{r_{k}^{2}+r_{b}^{2}}{\tilde{\sigma}^{2}(1-\mu_{b}^{2})}\right]\nn\\
		&
		\times \mathrm{I}_{0}\!\left(\frac{2r_{k}r_{b}}{\tilde{\sigma}^{2}(1-\mu_{b}^{2})}\right),
		\label{eq:ak_pdf}
	\end{align}
	where $\mathrm{I}_{0}(\cdot)$ denotes the modified Bessel function of the first kind 
	and zero order, and $\Delta_{b} = |\delta + \mu_{b}\tilde{\bh}_{b}^{\T}\mathbf{g}_{\mathrm{SIM}}|$.
	
	Given that $\Delta_{b}$ follows a Rician distribution, its probability density 
	function (PDF) is expressed as
	\begin{equation}
		f_{\Delta_{b}}(r_{b})
		=
		\frac{2r_{b}}{\tilde{\sigma}^{2}\mu_{b}^{2}}
		\exp\!\left[-\frac{r_{b}^{2}+|\delta|^{2}}{\tilde{\sigma}^{2}\mu_{b}^{2}}\right]
		\mathrm{I}_{0}\!\left(\frac{2r_{b}|\delta|}{\tilde{\sigma}^{2}\mu_{b}^{2}}\right).
		\label{eq:lambda_pdf}
	\end{equation}
	
	Conditioned on $\Delta_{b}$, the random variables $C_{1}, C_{2}, \ldots, C_{N}$ 
	within the $b$-th block are statistically independent. 
	Hence, the joint PDF of 
	$\{C_{k}\}$ given $\Delta_{b}$ can be expressed as
	\begin{align}
		&f_{C_{1},C_{2},\ldots,C_{N}|\Delta_{b}}
		(r_{1},r_{2},\ldots,r_{N}\,|\,r_{b})
		=
		\prod_{k\in\mathcal{K}_{b}}
		\frac{2r_{k}}{\tilde{\sigma}^{2}(1-\mu_{b}^{2})}\nn\\
		&\times	\exp\!\left[-\frac{r_{k}^{2}+r_{b}^{2}}
		{\tilde{\sigma}^{2}(1-\mu_{b}^{2})}\right]
		I_{0}\!\left(
		\frac{2r_{k}r_{b}}
		{\tilde{\sigma}^{2}(1-\mu_{b}^{2})}
		\right),
		\label{eq:joint_pdf_Ak}
	\end{align}
	where $\mathcal{K}_{b}$ denotes the set of indices associated with 
	the $b$-th block (i.e., $|\mathcal{K}_{b}|=L_{b}$).

	By combining  \eqref{eq:lambda_pdf} and \eqref{eq:joint_pdf_Ak}, 
	the joint cumulative distribution function (CDF) of 
	$\{C_{n}\}$ under the BDMA model can be written as \eqref{eq:cdf_An}, where $Q_{1}(\cdot,\cdot)$ is the first-order Marcum $Q$-function and 
	$\mathcal{K}_{b}$ denotes the set of ports within the $b$-th block. Finally, by setting all $r_{k}=\sqrt{\gamma_{\mathrm{th}}}$ and substituting into 
	\eqref{eq:cdf_An}, the overall outage probability of the SIM–FAS system is obtained as
	\eqref{eq:Pout_final}, where $\gamma_{\mathrm{th}}=(2^{R}-1)\sigma^{2}/P$ represents the   SNR threshold corresponding to the target data rate $R$.
	
	\begin{figure*}
		\begin{align}
			&\!\!\!\!\!\!\!\!\!\!\!\!	F_{ C_{1}, \ldots,C_{N}}(r_{1},\ldots,r_{N})
			\!=\!
			\prod_{b=1}^{B}\!\!
			\int_{0}^{\infty}
			\!\!\!		\frac{2r_{b}}{\tilde{\sigma}^{2}\mu_{b}^{2}}
			\!\exp\!\left[-\frac{r_{b}^{2}+|\delta|^{2}}{\tilde{\sigma}^{2}\mu_{b}^{2}}\right]
			\!\!\mathrm{I}_{0}\!\left(\frac{2r_{b}|\delta|}{\tilde{\sigma}^{2}\mu_{b}^{2}}\right)
			\!\!\!\prod_{k\in\mathcal{K}_{b}}
			\!\!\left[
			\!1 \!-\! Q_{1}\!\!\left(\!\!
			\sqrt{\!\frac{2}{\tilde{\sigma}^{2}(1-\mu_{b}^{2})}}r_{b},\,
			\sqrt{\!\frac{2}{\tilde{\sigma}^{2}(1-\mu_{b}^{2})}}r_{k}
			\!\!\right)
			\!\right]\!\!dr_{b}.
			\label{eq:cdf_An}\\
			P_{\text{out}}
			&	=
			\prod_{b=1}^{B}
			\int_{0}^{\infty}
			\frac{2r_{b}}{\tilde{\sigma}^{2}\mu_{b}^{2}}
			\exp\!\left[-\frac{r_{b}^{2}+|\delta|^{2}}{\tilde{\sigma}^{2}\mu_{b}^{2}}\right]
			\mathrm{I}_{0}\!\left(\frac{2r_{b}|\delta|}{\tilde{\sigma}^{2}\mu_{b}^{2}}\right)
			\left[
			1 - Q_{1}\!\left(
			\sqrt{\frac{2}{\tilde{\sigma}^{2}(1-\mu_{b}^{2})}}r_{b},\,
			\sqrt{\frac{2\gamma_{\mathrm{th}}}{\tilde{\sigma}^{2}(1-\mu_{b}^{2})}}
			\right)
			\right]^{L_{b}}dr_{b}.
			\label{eq:Pout_final}
		\end{align}
		\hrulefill
	\end{figure*}

	\section{Outage Probability Optimization}
	
This section minimizes the outage probability under SIM phase-shift constraints, leading to the following problem.
	
	\begin{subequations}\label{eq:subeqns}
		\begin{align}
			(\mathcal{P})~~&\min_{\theta_{\ell,m}} 	\;~			P_{\text{out}} \label{Maximization1} \\
			&\;\quad\;\;\;\;\;\;\;\!\!\!~\!		\mathbf{g}_{\mathrm{SIM}} = \left( \prod_{\ell=2}^{L} \boldsymbol{\Theta}_{\ell}\mathbf{W}_{\ell} \right)
			\boldsymbol{\Theta}_{1}\mathbf{w}_{1},
			\label{Maximization4} \\
			&~	\mathrm{\textcolor{black}{s.t.}}\;\;\;\;\;\;\!\!~\!	\boldsymbol{\Theta}_\ell = \mathrm{diag}\big(e^{j\theta_{\ell,1}}, e^{j\theta_{\ell,2}}, \dots, e^{j\theta_{\ell,M}}\big),
			\label{Maximization3} \\
			&	\;\quad\;\;\;\;\;\;\;\!\!\!~\!		\theta_{\ell,m}\in [0, 2\pi), l \in \mathcal{L}, m \in \{1,\dots,M\}	\label{Maximization8}.
		\end{align}
	\end{subequations}

Problem $(\mathcal{P})$ is nonconvex due to both the objective function and the constant-modulus constraint on $\boldsymbol{\Theta}_\ell$.

	\subsection{Proposed Algorithm}
	\textcolor{black}{The optimization is performed by iteratively updating the SIM phase shifts along the gradient direction of the outage probability, followed by projection onto the unit-modulus constraint set, which converges to a stationary point under appropriate step-size selection. In this work, a sufficiently small or diminishing step size ensures stability, while a properly tuned fixed step size achieves fast convergence within a few iterations.}\footnote{	\textcolor{black}{The adopted projected gradient approach follows similar principles to gradient-based optimization methods used in FAS-enabled systems (e.g., antenna position optimization and continuous relaxations of port selection) and reconfigurable surface designs \cite{Chai2022,Mei2024}. While the gradient expression is specific to the SIM structure, the iterative update and projection mechanism is general and widely applicable~\cite{Papazafeiropoulos2024a,Papazafeiropoulos2025,Papazafeiropoulos2024b}.}}  \textcolor{black}{The resulting projected gradient method is summarized in 
	Algorithm~\ref{Algoa1}. 
	The positive parameter $\bar{\mu}_{i}>0$  represents the step size along the gradient directions of the objective function.  	The feasible set associated with the unit-modulus constraint is defined as
	$
	\Theta_{l}= 
	\big\{	|\theta_{\ell,m}|=1,\; m=1,\ldots,M \big\}$. 		During each iteration, the newly computed update may fall outside 
	this feasible set. Therefore, it is projected back onto the corresponding 
	constraint manifold via the projection operator $P_{\Theta_{l}}(\cdot)$. 	 }
	
	\begin{algorithm}[th]
		\caption{Projected Gradient Ascent Method}
		\label{Algoa1}
		\begin{algorithmic}[1]
			\STATE \textbf{Input:} Initial value 
			$\theta_{\ell,0}$
			and step size $\mu_{i}>0$.
			\STATE \textbf{for} $i=1,2,\ldots$ \textbf{do}
			\STATE \quad 
			$\theta_{\ell,m}^{i+1} 
			= P_{\Theta_{l}}\!\left(
			\theta_{\ell,m}^{i}
			+ \bar{\mu}_{i}\nabla_{\theta_{\ell,m}}	P_{\text{out}}		
			\right)$
			\STATE \textbf{end for}
		\end{algorithmic}
	\end{algorithm}
	
	The analytical expression for the gradient 
	$\nabla_{\theta_{\ell,m}}	P_{\text{out}}$ is given by the following proposition in closed form.
	\begin{proposition}\label{prop1}
		The gradient of $	P_{\text{out}} $ with respect to $ \theta_{\ell,m}$  is obtained as
		\begin{align}
			\nabla_{\theta_{\ell,m}}	P_{\text{out}}
			&=	
			-\,\frac{2\alpha}{K+1}
			\!\!\left(\!\frac{\partial P_{\text{out}}}{\partial \tilde{\sigma}^{2}}\!\right) \!\!\Im\!\left\{
			e^{j\boldsymbol{\theta}_{\ell}} \odot \mathbf{g}_{\ell}
			\right\}\nn\\
			&-\, \sqrt{\frac{\alpha K}{K+1}} \!\!\left(\!\frac{\partial P_{\text{out}}}{\partial |\delta|}\!\right) \!\!
			\Im\!\Big\{e^{j\boldsymbol{\theta}_{\ell}} \odot \tilde{\mathbf{g}}_{\ell}^{(\delta)}
			\Big\},
			\label{eq:dPout_dtheta}
		\end{align}
		where $\frac{\partial P_{\text{out}}}{\partial \tilde{\sigma}^{2}}$, $\frac{\partial P_{\text{out}}}{\partial |\delta|}$, $ \mathbf{g}_{\ell}$, and $\mathbf{g}_{\ell}^{(\delta)} $ are provided in the proof.
	\end{proposition}
	\begin{proof}
		Please see Appendix~\ref{Prop1proof}.	
	\end{proof}
	
	Regarding, the complexity of Algorithm~\eqref{Algoa1}, it is linear in the number of BDMA blocks $B$, i.e., $\mathcal(O(B))$. Also, it depends on the integration during the calculation of $P_{\text{out}}$. 
	
	\section{Numerical Results}
	
Numerical simulations validate the proposed SIM–FAS framework. The SIM is placed at height $H=10$~m, with the user located $60$~m away. Each layer contains $M\in{16,32}$ meta-atoms of size $\lambda/2 \times \lambda/2$, with total SIM thickness $T_{\mathrm{SIM}}=5\lambda$. Large-scale fading follows~\cite[Eq.~(16)]{An2023} with exponent $b=3.5$ and Rician factor $K=2$, while noise power is $\sigma^2=-96$~dBm. The FAS span is set to $W=5$, and the target rate is $6$~bps/Hz. Port correlation is fixed to $\mu_b^2=0.97$ for all blocks~\cite{RamirezEspinosa2024}.

	 \begin{figure*}[t]
		\begin{minipage}{0.33\textwidth}
			\centering
			\captionbox{Outage probability versus $P$, varying $M$.\label{Fig01}}{	\includegraphics[trim=0cm -1.2cm 0cm 0.1cm, clip=true, width=2.2in]{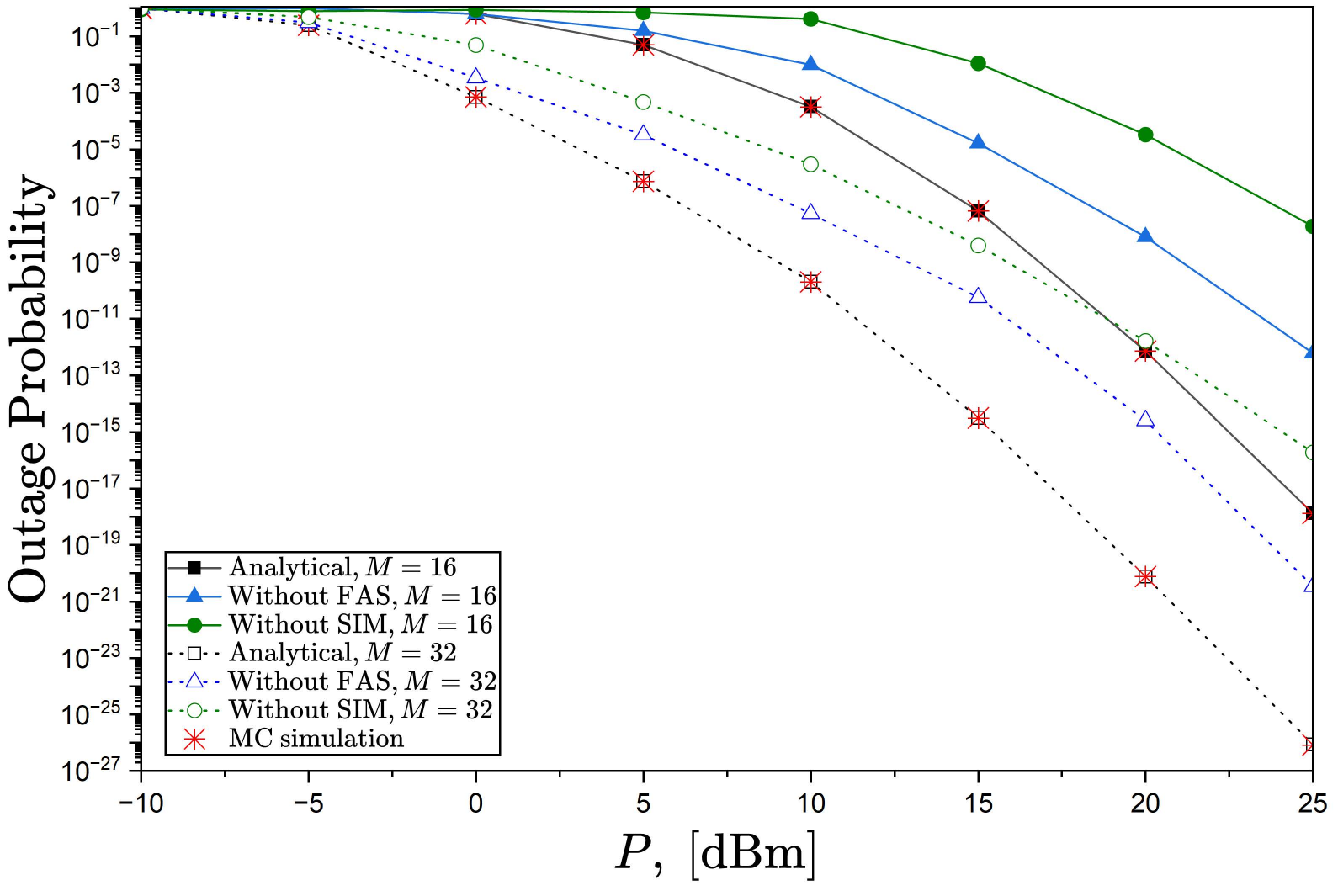}}
			\vspace*{-0.2cm}
		\end{minipage}
		\begin{minipage}{0.33\textwidth}
			\centering
			\captionbox{Outage probability versus $P$, varying $L$.\label{Fig02}}{	\includegraphics[trim=0cm 0.0cm 0cm 0.0cm, clip=true,height=1.6in, width=2.2in]{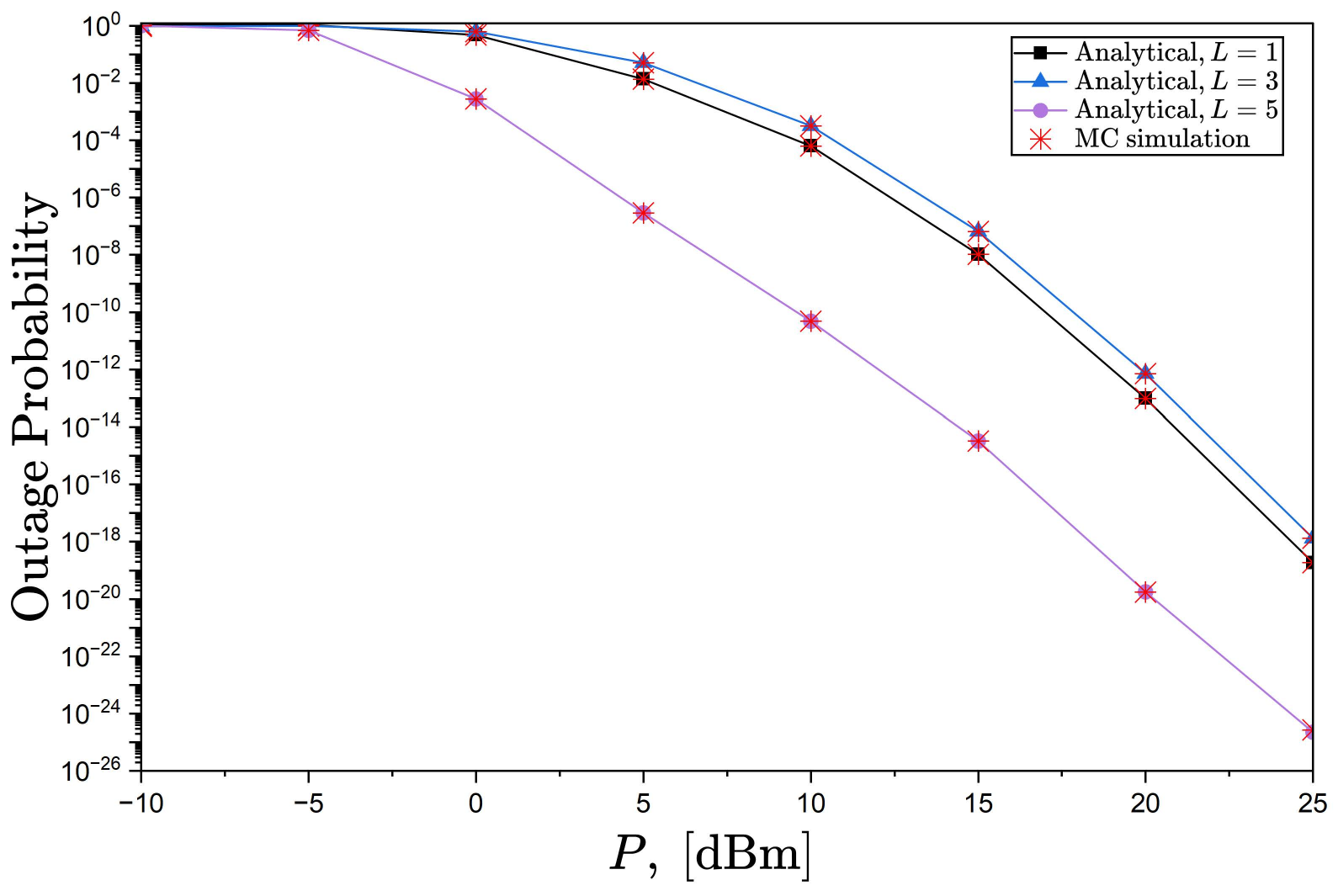}}
			\\ 
			\vspace*{-0.2cm}
		\end{minipage}
		\begin{minipage}{0.33\textwidth}
			\centering
			\captionbox{Outage probability versus $N$, varying $P$.\label{Fig03}}{	\includegraphics[trim=0cm 0cm 0cm 0.1cm, clip=true, height=1.6in, width=2.1in]{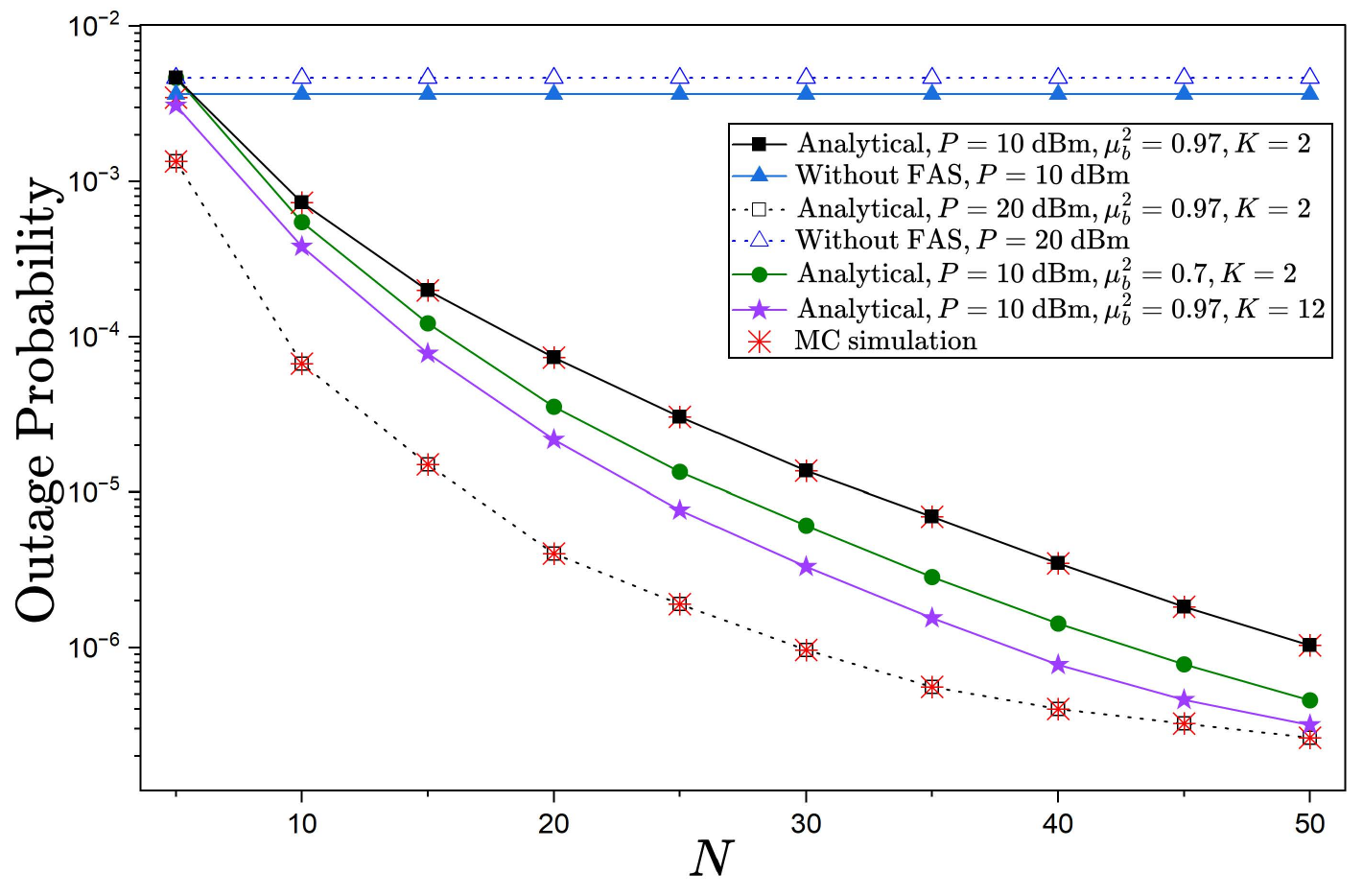}}
			\\ 
			\vspace*{-0.2cm}
		\end{minipage}
		\label{Fig22}
	\end{figure*}

Fig.~\ref{Fig01} shows the outage probability versus transmit power $P$ for $N=50$ and $L=3$. The close agreement between analytical and Monte Carlo results validates the BDMA-based analysis~\cite{RamirezEspinosa2024}. Increasing $M$ improves performance due to higher array gain, while comparisons with the`` Without SIM'' and ``Without FAS'' cases demonstrate the substantial outage reduction achieved by jointly employing SIM and FAS.
	
		Fig.~\ref{Fig02} depicts the impact of the number of SIM surfaces $L$ on the outage probability for $M=16$, where an increase in $L$ yields a decrease of the outage probability due to higher wave-domain degrees of freedom, enabling stronger constructive combining and better focusing of the transmitted energy toward the user.

Fig.~\ref{Fig03} shows the outage probability versus the number of ports 
$N$ for different 
$P$. In \eqref{eq:Pout_final},  $N$ affects performance through the BDMA block sizes 
$L_b$, which determine the diversity order. Thus, increasing 
$N$ improves performance, while the ``Without FAS'' case remains unaffected by 
$N$. \textcolor{black}{In addition, we observe that higher intra-block correlation $\mu_b$ reduces effective diversity, slowing the outage decay with $N$, while lower $\mu_b$ enhances performance. Moreover, larger $K$ strengthens the LoS component and improves coherent combining, reducing outage, whereas smaller $K$ increases fading severity.}
\textcolor{black}{Note that the diversity gain is governed by the effective channel rank rather than the number of ports 
	$N$ \cite{Zhu2025a}. Under BDMA, this corresponds to the number of independent blocks, leading to diminishing returns as correlation increases as shown by the figure.}

	\section{Conclusion}
This work proposes a novel SIM–FAS framework with tractable outage analysis via BDMA. A closed-form outage expression was derived and used for SIM optimization. Numerical results validate the analysis and show clear gains over benchmarks. Future work may extend the model to multi-antenna BS scenarios.
	\begin{appendices}
		
		\section{Proof of Proposition~\ref{prop1}}\label{Prop1proof}	
		
		Starting from \eqref{eq:Pout_final}, the outage probability can be expressed as
		\begin{equation}
			P_{\text{out}}(\tilde{\sigma}^{2}, |\delta|) = \prod_{b=1}^{B} J_{b}(\tilde{\sigma}^{2}, |\delta|).
		\end{equation}
		
		where
		\begin{align}
			J_{b}(\tilde{\sigma}^{2})
			&= \int_{0}^{\infty}
			\underbrace{
				\frac{2r_{b}}{\tilde{\sigma}^{2}\mu^{2}}
				\exp\!\left(-\frac{r_{b}^{2}+|\delta|^{2}}{\tilde{\sigma}^{2}\mu^{2}}\right)		I_{0}\!\left(\frac{2r_{b}|\delta|}{\tilde{\sigma}^{2}\mu^{2}}\right)
			}_{\triangleq\, f_{b}(r_{b};\,\tilde{\sigma}^{2},  |\delta|)}
			\nn\\
			&	\bigg[
			1 - Q_{1}\!\left(
			\sqrt{\tfrac{2r_{b}}{\tilde{\sigma}^{2}(1-\mu^{2})}},
			\sqrt{\tfrac{2\gamma_{\mathrm{th}}}{\tilde{\sigma}^{2}(1-\mu^{2})}}
			\right)
			\bigg]^{L_{b}}
			dr_{b}.
			\label{eq:Jb_sigma}
		\end{align}
		Therefore, by the chain rule, we obtain
		\begin{equation}
			\frac{\partial P_{\text{out}}}{\partial \theta_{\ell,m}} =
			\frac{\partial P_{\text{out}}}{\partial \tilde{\sigma}^{2}}
			\frac{\partial \tilde{\sigma}^{2}}{\partial \theta_{\ell,m}}
			+
			\frac{\partial P_{\text{out}}}{\partial |\delta|}
			\frac{\partial |\delta|}{\partial \theta_{\ell,m}}.\label{eq3}
		\end{equation}

		Firstly, differentiating $P_{\text{out}}$ with respect to $\tilde{\sigma}^{2}$ yields
		\begin{equation}
			\frac{\partial P_{\text{out}}}{\partial \tilde{\sigma}^{2}}
			= P_{\text{out}}
			\sum_{b=1}^{B}
			\frac{1}{J_{b}}
			\frac{\partial J_{b}}{\partial \tilde{\sigma}^{2}},
			\label{eq:dPout_dsigma}
		\end{equation}
		where the term $\partial J_{b}/\partial \tilde{\sigma}^{2}$ is obtained by differentiating 
		under the integral sign, since $\tilde{\sigma}^{2}$ appears in the exponential, 
		in the argument of the modified Bessel function $I_{0}(\cdot)$, 
		and in the parameters of the Marcum $Q_{1}$ function.
		
		The effective variance is related to the SIM  as
		\begin{equation}
			\tilde{\sigma}^{2} = \frac{\alpha}{K+1}\,\|\mathbf{g}_{\mathrm{SIM}}(\boldsymbol{\Theta})\|_{2}^{2}.
		\end{equation}
		Hence,
		\begin{equation}
			\frac{\partial \tilde{\sigma}^{2}}{\partial \boldsymbol{\Theta}_{\ell}}
			= \frac{\alpha}{K+1}\,
			\frac{\partial \|\mathbf{g}_{\mathrm{SIM}}\|_{2}^{2}}{\partial \boldsymbol{\Theta}_{\ell}},\label{der1}
		\end{equation}
		where $
		d\|\mathbf{g}_{\mathrm{SIM}}\|_{2}^{2} = 2\,\Re\!\{\mathbf{g}_{\mathrm{SIM}}^{H} d\mathbf{g}_{\mathrm{SIM}}\}$.
		
		For a SIM with a single RF feed, the transmitted field 
		can be written as
		\begin{equation}
			\mathbf{g}_{\mathrm{SIM}} = \mathbf{D}_{\ell}\,\boldsymbol{\Theta}_{\ell}\,\mathbf{u}_{\ell},\label{eq1}
		\end{equation}
		where
		\begin{align}
			\mathbf{D}_{\ell} &\triangleq \boldsymbol{\Theta}_{L}
			\mathbf{W}_{L}\boldsymbol{\Theta}_{L-1}\mathbf{W}_{L-1}\cdots,\boldsymbol{\Theta}_{\ell+1}\mathbf{W}_{\ell+1} \\
			\mathbf{u}_{\ell} &\triangleq 
			\mathbf{W}_{\ell}\boldsymbol{\Theta}_{\ell-1}\mathbf{W}_{\ell-1}\cdots\boldsymbol{\Theta}_{1}\mathbf{w}_{1}.
		\end{align}
		Given that  $\boldsymbol{\Theta}_{\ell} = \mathrm{diag}\!\left(e^{j\theta_{\ell,1}},\ldots,
		e^{j\theta_{\ell,M}}\right)$, we obtain
		\begin{equation}
			\frac{\partial \mathbf{g}_{\mathrm{SIM}}}{\partial \theta_{\ell,m}}
			= j e^{j\theta_{\ell,m}} u_{\ell,m}\, \mathbf{d}_{\ell,m},\label{eq2}
		\end{equation}
		where $\mathbf{d}_{\ell,m}$ is the $m$th column of $\mathbf{D}_{\ell}$. 
		Therefore,
		\begin{align}
			\frac{\partial \|\mathbf{g}_{\mathrm{SIM}}\|_{2}^{2}}{\partial \theta_{\ell,m}}
			&= 2\,\Re\!\left\{\mathbf{g}_{\mathrm{SIM}}^{H} \frac{\partial \mathbf{g}_{\mathrm{SIM}}}{\partial \theta_{\ell,m}}\right\}\nn\\
			&		= -2\,\Im\!\left\{
			e^{j\theta_{\ell,m}} u_{\ell,m}\,\mathbf{g}_{\mathrm{SIM}}^{H}\mathbf{d}_{\ell,m}
			\right\}.\label{der2}
		\end{align}
		By substituting \eqref{der2} into \eqref{der1},  we obtain
		\begin{align}
			\frac{\partial \tilde{\sigma}^{2}}{\partial \theta_{\ell,m}}
			&= -\,\frac{2\alpha}{K+1}\,
			\Im\!\left\{
			e^{j\theta_{\ell,m}} u_{\ell,m}\,\mathbf{g}_{\mathrm{SIM}}^{H}\mathbf{d}_{\ell,m}
			\right\}\nn\\
			&=-\,\frac{2\alpha}{K+1}
			\Im\!\left\{
			e^{j\boldsymbol{\theta}_{\ell}} \odot \mathbf{g}_{\ell}
			\right\},
			\label{eq:dsigma_dtheta}
		\end{align}
		where $	\mathbf{g}_{\ell} \triangleq 
		\mathrm{diag}\!\left(\mathbf{D}_{\ell}^{H} \mathbf{g}_{\mathrm{SIM}}\right)
		\odot \mathbf{u}_{\ell}$ with $\odot$ denoting the Hadamard (elementwise) product.  Note that $e^{j\boldsymbol{\theta}_{\ell}}$ denotes the elementwise complex exponential 
		of the phase vector $\boldsymbol{\theta}_{\ell}=
		[\theta_{\ell,1},\,\theta_{\ell,2},\,\ldots,\,\theta_{\ell,M}]^{\mathsf T}$, and $\Im\{\cdot\}$ denotes the 
		elementwise imaginary part.
		
		By combining \eqref{eq:dPout_dsigma} and \eqref{eq:dsigma_dtheta}, we obtain the first part of $ \eqref{eq3}$.
		
		The remaining term, $\partial P_{\text{out}}/\partial \tilde{\sigma}^{2}$, 
		is computed from \eqref{eq:dPout_dsigma} as
		\begin{align}
			\frac{\partial J_{b}}{\partial \tilde{\sigma}^{2}}
			&= 
			\int_{0}^{\infty}
			\frac{\partial}{\partial \tilde{\sigma}^{2}}
			\Big(
			f_{b}(r_{b};\,\tilde{\sigma}^{2})
			\big[1-Q_{1}(\cdot,\cdot)\big]^{L_{b}}
			\Big)
			dr_{b},
		\end{align}
		where the inner derivatives are obtained using the relations
		\begin{align*}
			&	\frac{d}{d\tilde{\sigma}^{2}}\!\left(\frac{1}{\tilde{\sigma}^{2}}\right)
			= -\frac{1}{(\tilde{\sigma}^{2})^{2}},~ 
			\frac{d}{d\tilde{\sigma}^{2}}e^{-c/\tilde{\sigma}^{2}}
			= \frac{c}{(\tilde{\sigma}^{2})^{2}} e^{-c/\tilde{\sigma}^{2}},\\
			\frac{d}{d\tilde{\sigma}^{2}} &I_{0}\!\left(\frac{\kappa}{\tilde{\sigma}^{2}}\right)
			= -\frac{\kappa}{(\tilde{\sigma}^{2})^{2}} I_{1}\!\left(\frac{\kappa}{\tilde{\sigma}^{2}}\right),~
			\frac{\partial Q_{1}}{\partial a}
			= b e^{-(a^{2}+b^{2})/2} I_{1}(ab),\\
			&\!\!\!\!\!\!\!\!\frac{\partial Q_{1}}{\partial b}
			= -b e^{-(a^{2}+b^{2})/2} I_{0}(ab),~\frac{\partial a}{\partial \tilde{\sigma}^{2}}
			= -\frac{a}{2\tilde{\sigma}^{2}},~
			\frac{\partial b}{\partial \tilde{\sigma}^{2}}
			= -\frac{b}{2\tilde{\sigma}^{2}}.
		\end{align*}
		Substituting these expressions into $\partial J_{b}/\partial \tilde{\sigma}^{2}$ 
		provides this partial derivative.

		Secondly, differentiating under the integral sign with respect to $|\delta|$ gives
		\begin{align}
			&	\frac{\partial P_{\mathrm{out}}}{\partial |\delta|}
			= 
			P_{\mathrm{out}}
			\sum_{b=1}^{B} 
			\frac{1}{J_b}s
			\int_{0}^{\infty} 
			f_b(r_b; \tilde{\sigma}^2, |\delta|)
			\Bigg[
			\frac{2|\delta|}{\tilde{\sigma}^2 \mu_b^2}
			+
			\frac{2r_b}{\tilde{\sigma}^2 \mu_b^2}
			\frac{I_1(u)}{I_0(u)}
			\Bigg]\nn\\
			&\times
			[1 - Q_1(a,b)]^{L_b}
			\, dr_b,
			\label{eq:A_dPout_ddelta}
		\end{align}
		where $u = \frac{2 r_b |\delta|}{\tilde{\sigma}^2 \mu_b^2}$.

		For the  dependence of $|\delta|$ on $\theta_{\ell,m}$, using similar steps to  \eqref{eq1}-\eqref{eq2}, we obtain
		\begin{align}
			\frac{\partial|\delta|}{\partial\boldsymbol{\theta}_\ell}
			= -\sqrt{\frac{\alpha K}{K+1}}\;
			\Im\!\Big\{e^{j\boldsymbol{\theta}_{\ell}} \odot \tilde{\mathbf{g}}_{\ell}^{(\delta)}
			\Big\}
			\label{eq:dabsdelta_dtheta},
		\end{align}
		where $\boldsymbol{s}_{\ell} \triangleq \mathbf{D}_{\ell}^{\mathrm{T}} \, \bar{\mathbf{h}}_{k} \in \mathbb{C}^{M \times 1}
		$ and $\mathbf{g}_{\ell}^{(\delta)} \triangleq \mathbf{u}_{\ell} \odot \boldsymbol{s}_{\ell}$. This completes the derivation of the second derivative in \eqref{eq3}, and completes the proof.
		
	\end{appendices}
	
	\bibliographystyle{IEEEtran}

	\bibliography{IEEEabrv,bibl}
\end{document}